\documentclass[12pt, draftclsnofoot, onecolumn]{IEEEtran}

\ifCLASSINFOpdf
  \usepackage[pdftex]{graphicx}
  \graphicspath{{/Users/zahmed/mySharedFolder/per/PhD/Latex/IEEEtran/pdf/}{/Users/zahmed/mySharedFolder/per/PhD/Latex/IEEEtran/jpeg/}}
  \DeclareGraphicsExtensions{.pdf,.jpeg,.png}
\else
  \DeclareGraphicsExtensions{.eps}
\fi
\usepackage{algorithm}
\usepackage{algpseudocode}
\usepackage{tabu}
\usepackage{setspace}
\usepackage{makecell}
\usepackage{epsfig}
\usepackage[usenames, dvipsnames]{color}
\usepackage[roman]{parnotes}

\usepackage{amsthm,amsmath,amssymb}
\usepackage{mathtools}
\usepackage{breqn}
\usepackage{amsfonts}
\usepackage{amssymb}
\usepackage{mathrsfs}
\usepackage{xparse}
\usepackage{multirow}
\usepackage{chngcntr}
\usepackage{fancyhdr}
\DeclareMathOperator{\tr}{tr}
\allowdisplaybreaks

\begin{document}
%
\title{Optimal Bit Allocation Variable-Resolution ADC for Massive MIMO}
\author{\IEEEauthorblockN{I. Zakir Ahmed$^\star$, Hamid Sadjadpour$^\star$, and Shahram Yousefi$^\ast$\\}
\IEEEauthorblockA{$^\star$ Department of Electrical and Computer Engineering, UC Santa Cruz. $^\ast$Department of Electrical and Computer Engineering, Queen's University, Canada}}


%


\maketitle
\thispagestyle{fancy}
\vspace{-1.0in}
\begin{abstract}
In this paper, we derive an optimal ADC bit-allocation (BA) condition for a Single-User (SU) Millimeter wave (mmWave) Massive Multiple-Input Multiple-Output (Ma-MIMO) receiver equipped with variable-resolution ADCs under power constraint with the following criteria: (i) Minimizing the Mean Squared Error (MSE) of the received, quantized and combined symbol vector and (ii) Maximizing the capacity of the SU mmWave Ma-MIMO channel encompassing hybrid precoder and combiner. Optimal BA under both criteria results the same.
We jointly design the hybrid combiner  based on the SVD of the channel. We demonstrate improvement of  the proposed optimal BA over the BA based on Minimization of the Mean Square Quantization Error (MSQE). Using Monte-Carlo simulations, it is shown that the MSE and capacity performance of the proposed BA is very close to that of the Exhaustive Search (ES). 
The computational complexity of the proposed techniques are compared with ES and MQSE BA algorithms.
\end{abstract}


%

%
\IEEEpeerreviewmaketitle

\newcommand{\Xmatrix}{
\begin{bmatrix}
\ddots  & 0     & 0 \\
0  & \frac{1}{{\sigma_i^2}} & 0 \\
0 & 0 & \ddots
\end{bmatrix}}
\newcommand{\Ymatrix}{
\begin{bmatrix}
\ddots  & 0  & 0 \\
0  & \frac{f(b_i)l_i}{\big(1-f(b_i)\big) \sigma_i^2} & 0 \\
0 & 0 & \ddots
\end{bmatrix}}
\newcommand{\Zmatrix}{
\begin{bmatrix}
\ddots  & 0  & 0 \\
0  & \frac{\sigma_i^2}{\sigma_n^2 + \frac{f(b_i)l_i}{\big(1-f(b_i)\big)}} & 0 \\
0 & 0 & \ddots
\end{bmatrix}}
\newcommand{\ZMmatrix}{
\begin{bmatrix}
\ddots  & 0  & 0 \\
0  & \frac{\sigma_i^2}{\sigma_n^2 + \frac{f(b_i)l_i}{\big(1-f(b_i)\big)}} + \frac{1}{p} & 0 \\
0 & 0 & \ddots
\end{bmatrix}}
\newcommand{\Fmatrix}{
\begin{bmatrix}
\ddots  & 0  & 0 \\
0  & 10 & 0 \\
0 & 0 & \ddots 
\end{bmatrix}}
\newcommand{\InvFmatrix}{
\begin{bmatrix}
\ddots  & 0 & 0 \\
0  & f(b_i)\big(1-f(b_i)\big)l_i & 0 \\
0 & 0 & \ddots
\end{bmatrix}}

\section{Introduction}\label{Intro}
Massive Multiple-Input Multiple-Output (Ma-MIMO) technology is a key feature in consideration for next generation of wireless communication standards. It is being considered both at sub-6Ghz  and millimeter wave (mmWave) frequencies $\cite{5GBackHaul, SigProc}$. In both scenarios, a large number of antennas help to increase the capacity of the system or increase the Energy-Efficiency (EE). With favorable channel conditions, a combination of both can be achieved. Hybrid precoding and combining can be used to capture the potential advantages. As such, precoders and combiners both in analog and digital domains, are adapted to the changing channel conditions $\cite{SigProc,mmPreCom}$.

The Single-User (SU) Ma-MIMO typically found in the deployment of back-haul wireless interconnects between the Base Stations (BS) \cite{5GBackHaul,5GBackHaul2}, can exploit large-bandwidth to provide for large back-haul traffic with multiple streams or data channels over a single link. The large number of RF paths increases the cost of the RF components in addition to power consumption. By splitting the precoding and combining between analog and digital domains (hybrid precoding and combining), the number of RF paths can be reduced considerably as compared to the number of transmit and receive antennas. The ADCs operating at such large bandwidths consume large amount of power \cite{5GBackHaul,SigProc,Rangan}. In addition to power consumption, high resolution ADCs operating at high sampling frequencies produce huge amounts of data that is difficult to handle. Using very low bit resolution ADCs (1-2 bits) is a popular approach for massive MIMO receiver architectures to mitigate large power demands \cite{Risi,Jmo}. However this comes at a cost of performance especially at medium to high SNR regimes. It has been shown that adopting variable-resolution ADCs in Ma-MIMO receivers improves the performance of the communication link both in  terms of Mean Squared Error (MSE) and overall capacity under a receiver power constraint. In addition, the ADC bit-allocation (BA) has been found to depend on the channel conditions \cite{VarBitAlloc,Zakir1,Zakir2,Zakir3}. 
\vspace{-0.6in}
\subsection{Previous Works} \label{pwork}
Low resolution ADC MIMO receiver architectures using 1-bit and few n-bit  have been studied extensively over the last few years \cite{Risi,Jmo,Mezghani,Muris,HybArchCap,Uplink}. The 1-bit ADC receiver architecture \cite{Muris,HybArchCap,Uplink} improves EE, however at a cost of performance at medium to high SNR regimes. In \cite{Jmo}, it is shown that despite improved deployment cost and EE, there is considerable rate loss in the medium to high SNR regimes with 1-bit ADC architectures. In \cite{SvenCap}, it is shown that by a small increase in the resolution of ADCs (eg., 3-bits) on all RF paths, significant performance gains can be achieved for a broad range of system parameters. In \cite{Muris}, the authors choose a realistic system model and setup (reflecting the hardware implementation) and perform a parametric analysis as a function of ADC resolution. The EE is shown to be maximized at intermediate ADC resolutions typically between 4-8 bits .

All mentioned papers above use equal-bit-resolution ADCs on the receiver's RF paths. Since the resolutions of ADCs are fixed and low, the receiver is power efficient while suffering from performance loss. Hence, an optimal performance-EE trade-off is to be obtained. It has been shown \cite{VarBitAlloc,VarBitAllocJour,Zakir1,Zakir2,Zakir3} that the ADC resolutions on each RF path need not be equal; they rather need to be adjusted for a given power budget and channel conditions. Thus, employing variable-resolution ADCs on the receiver's RF paths can be advantageous. The BA schemes can typically be updated over one to a few coherence times. Choi et. al. \cite{VarBitAlloc,VarBitAllocJour,VarBitEner} derived an ADC BA mechanism based on the minimization of the Mean Square Quantization Error (MSQE) under  receiver power constraint. 
\subsection{Our Contribution}\label{cwork}
We derive optimal ADC BA under receiver power constraint for two criteria: (i) Minimizing the MSE of the received, quantized, and combined symbol vector and (ii) Maximizing the capacity of the SU mmWave Ma-MIMO channel. We show that the optimal BA conditions for both criteria coincide owing to the relationship between the capacity and the  Cramer Rao Lower Bound (CRLB) that we establish in Section $\ref{ba}$.\\ 
\indent
\textit{i) MSE minimization criterion:} The expression for MSE is derived and shown that it approaches the CRLB. Minimizing the MSE imposes multiple constraints and hence alternatively we minimize the CRLB with respect to the BA matrix. The BA matrix facilitates variable-bit allocation on the receiver's RF paths. In doing so, we arrive at the conditions for hybrid combiners and a simple algorithm for BA. \\ 
\indent
\textit{ii) Capacity maximization criterion:} The capacity expression for a given SU mmWave channel as a function of BA and hybrid combining matrices is derived. We show that the capacity is a function of the CRLB derived  for MSE minimization. We devise a BA solution based on maximizing the capacity and arrive at exactly the same conditions as those from the minimization of CRLB \cite{Zakir2}. 
\indent


The column vectors are represented as boldface small letters and matrices as boldface uppercase letters. The primary diagonal of a matrix is denoted as $\text{diag}(\cdot)$ and all expectations $E[\cdot]$ are over the random variable $\bold{n}$, which is an AWGN vector, i.e., $E[\cdot] = E_{\bold{n}}[\cdot]$. The multivariate normal distribution with mean $\boldsymbol{\mu}$ and covariance $\boldsymbol{\varphi}$ is denoted as $\mathcal{N}(\boldsymbol{\mu},\boldsymbol{\varphi})$ and $\mathcal{CN}(\bold{0},{\boldsymbol{\varphi}})$ denotes a multivariate complex-valued circularly-symmetric Gaussian distribution. Frobenius norm of a matrix $\bold{A}$ is shown as ${\lVert\bold{A}\rVert }_F$, the trace as $\tr{(\bold{A})}$ and the identity matrix of size $N$ as $\bold{I}_N$. The term $h(\bold{x})$ defines the differential entropy of a continuous random variable $\bold{x}$. The superscripts $T$ and $H$ denote transpose and Hermitian transpose, respectively. 

This paper is organized as follows. Section \ref{sigmod} describes the system model and parameters. Section \ref{PreCoderLb} describes the precoder design and Section \ref{ba} derives the optimal BA conditions for the two scenarios mentioned above. Section \ref{comb_dsgn} details the optimal hybrid combiner structure and design. In Section \ref{Sim}, we present the simulation results, test setup, and computational complexity analysis, followed by the conclusion in Section \ref{conc}.

\section{Signal Model}\label{sigmod}
The signal model for a typical SU Ma-MIMO transceiver encompassing a hybrid precoding and combining is shown in Figure \ref{fig:Fig1new.pdf}. ${\bold{F}_D}$ and ${\bold{F}_A}$ denote the digital and analog precoders, respectively. Similarly,   ${\bold{W}_D^H}$ and ${\bold{W}_A^H}$ represent the digital and analog combiners, respectively. The vector $\bold{x}$ is an $Ns\times1$ transmitted signal vector with unit average power. Let $N_{rt}$ and $N_{rs}$ denote the number of RF Chains at the transmitter and  receiver, respectively. Also, $N_t$ and $N_r$ represent the number of transmit and receive antennas, respectively. The channel matrix $\bold{H} = \big[ h_{ij} \big]$ is an \begin{math}(N_r\times N_t)\end{math}  matrix representing the line of sight mmWave MIMO channel with properties defined in $\cite{rapaport}$ (chapter 3, pages 99-125).

\begin{figure}[h]
\centering
\setlength{\unitlength}{3947sp}%
\begingroup\makeatletter\ifx\SetFigFont\undefined%
\gdef\SetFigFont#1#2#3#4#5{%
  \reset@font\fontsize{#1}{#2pt}%
  \fontfamily{#3}\fontseries{#4}\fontshape{#5}%
  \selectfont}%
\fi\endgroup%
\begin{picture}(3589,1349)(136,-861)
\put(3151,-511){\makebox(0,0)[lb]{\smash{{\SetFigFont{8}{9.6}{\familydefault}{\mddefault}{\updefault}{\color[rgb]{0,0,0}$\bold{r}$}%
}}}}
\thinlines
{\color[rgb]{0,0,0}\put(345, 75){\framebox(730,389){}}
}%
{\color[rgb]{0,0,0}\put(2341, 75){\framebox(729,389){}}
}%
{\color[rgb]{0,0,0}\put(2341,-849){\framebox(729,389){}}
}%
{\color[rgb]{0,0,0}\put(1367,-849){\framebox(730,389){}}
}%
{\color[rgb]{0,0,0}\put(345,-849){\framebox(730,389){}}
}%
\thicklines
{\color[rgb]{0,0,0}\put(1075,270){\vector( 1, 0){292}}
}%
{\color[rgb]{0,0,0}\put(2097,270){\vector( 1, 0){244}}
}%
{\color[rgb]{0,0,0}\put(3070,270){\vector( 1, 0){291}}
}%
{\color[rgb]{0,0,0}\put(3411,-120){\vector( 0, 1){341}}
}%
{\color[rgb]{0,0,0}\put(3507,270){\line( 1, 0){196}}
\put(3703,270){\line( 0,-1){925}}
\put(3703,-655){\vector(-1, 0){633}}
}%
{\color[rgb]{0,0,0}\put(2341,-655){\vector(-1, 0){244}}
}%
{\color[rgb]{0,0,0}\put(1367,-655){\vector(-1, 0){292}}
}%
{\color[rgb]{0,0,0}\put(199,270){\vector( 1, 0){146}}
}%
{\color[rgb]{0,0,0}\put(345,-655){\vector(-1, 0){146}}
}%
\thinlines
{\color[rgb]{0,0,0}\put(1367, 75){\framebox(730,389){}}
}%
\put(2551,-661){\makebox(0,0)[lb]{\smash{{\SetFigFont{8}{9.6}{\familydefault}{\mddefault}{\updefault}{\color[rgb]{0,0,0}$\bold{W}_A^H$}%
}}}}
\put(1501,-661){\makebox(0,0)[lb]{\smash{{\SetFigFont{8}{9.6}{\familydefault}{\mddefault}{\updefault}{\color[rgb]{0,0,0}$\text{Q}_{\bold{b}} \big( {\bold{z}} \big)$}%
}}}}
\put(526,-661){\makebox(0,0)[lb]{\smash{{\SetFigFont{8}{9.6}{\familydefault}{\mddefault}{\updefault}{\color[rgb]{0,0,0}$\bold{W}_D^H$}%
}}}}
\put(3376,-211){\makebox(0,0)[lb]{\smash{{\SetFigFont{8}{9.6}{\familydefault}{\mddefault}{\updefault}{\color[rgb]{0,0,0}$\bold{n}$}%
}}}}
\put(2176,389){\makebox(0,0)[lb]{\smash{{\SetFigFont{8}{9.6}{\familydefault}{\mddefault}{\updefault}{\color[rgb]{0,0,0}$\bold{\tilde{x}}$}%
}}}}
\put(151,389){\makebox(0,0)[lb]{\smash{{\SetFigFont{8}{9.6}{\familydefault}{\mddefault}{\updefault}{\color[rgb]{0,0,0}$\bold{x}$}%
}}}}
\put(2626,239){\makebox(0,0)[lb]{\smash{{\SetFigFont{8}{9.6}{\familydefault}{\mddefault}{\updefault}{\color[rgb]{0,0,0}$\bold{H}$}%
}}}}
\put(1576,239){\makebox(0,0)[lb]{\smash{{\SetFigFont{8}{9.6}{\familydefault}{\mddefault}{\updefault}{\color[rgb]{0,0,0}$\bold{F}_A$}%
}}}}
\put(601,239){\makebox(0,0)[lb]{\smash{{\SetFigFont{8}{9.6}{\familydefault}{\mddefault}{\updefault}{\color[rgb]{0,0,0}$\bold{F}_D$}%
}}}}
\put(1201,-511){\makebox(0,0)[lb]{\smash{{\SetFigFont{8}{9.6}{\familydefault}{\mddefault}{\updefault}{\color[rgb]{0,0,0}$\bold{\tilde{y}}$}%
}}}}
\put(2176,-511){\makebox(0,0)[lb]{\smash{{\SetFigFont{8}{9.6}{\familydefault}{\mddefault}{\updefault}{\color[rgb]{0,0,0}$\bold{z}$}%
}}}}
\put(151,-511){\makebox(0,0)[lb]{\smash{{\SetFigFont{8}{9.6}{\familydefault}{\mddefault}{\updefault}{\color[rgb]{0,0,0}$\bold{y}$}%
}}}}
\put(3376,239){\makebox(0,0)[lb]{\smash{{\SetFigFont{9}{10.8}{\rmdefault}{\mddefault}{\updefault}{\color[rgb]{0,0,0}+}%
}}}}
{\color[rgb]{0,0,0}\put(3411,288){\circle{136}}
}%
\end{picture}%
\caption{Signal Model.}
\label{fig:Fig1new.pdf}
\end{figure}
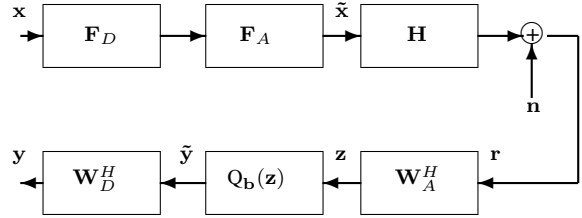

The transmitted signal $\bold{\tilde{x}}$ and the received signal $\bold{r}$ are thus known as ${\bold{\tilde{x}}} = {\bold{F}_A}{\bold{F}_D}{\bold{x}}, \text{   }{\bold{r}} = {\bold{H}}{\bold{\tilde{x}}}+{\bold{n}}$.
Here, ${\bold{n}}$ is an $N_r\times1$ noise vector of independent and identically distributed (i.i.d) complex Gaussian random variables such that ${\bold{n}} \sim \mathcal{CN}(\bold{0},{\sigma_n^2}{\bold{I}_{N_r}})$. The received symbol vector $\bold{r}$ is analog-combined with ${\bold{W}_A^H}$ to get ${\bold{z}} = {\bold{W}_A^H}{\bold{r}}$  and later  digitized using a variable-bit quantizer \cite{VarBitAlloc,Zakir1} to produce ${\bold{\tilde{y}}} = \text{Q}_{\bold{b}} \big( {\bold{z}} \big) = \bold{W}_{\alpha}\big( {\bold{b}} \big){\bold{z}}+{\bold{n_q}}$. This  signal is later combined using the digital combiner ${\bold{W}_D^H}$ to produce the output signal ${\bold{y}} = {\bold{W}_D^H}{\bold{\tilde{y}}}$. The quantizer is modeled as an Additive Quantization Noise Model (AQNM) $\cite{Rangan,Uplink}$. 
Here $\bold{b}=[b_1 b_2 b_3 .... b_N]^T$ is a vector whose entries $b_i$ indicate the number of bits $b_i$ (on both I and Q channels) that are allocated to the ADC on RF path $i$. The vector $\bold{n}_q$ has $ \mathcal{CN}(\bold{0},{\bold{D}_q^2})$ distribution and is uncorrelated with $\bold{z}$ \cite{VarBitAlloc}.

Hence, the relationship between the transmitted signal vector $\bold{x}$ and the received symbol vector $\bold{y}$ at the receiver is given by
\begin{equation}\label{eq5a}
{\bold{y}} = {\bold{W}_D^H}{\bold{W}_{\alpha}}{\big( {\bold{b}} \big)}{\bold{W}_A^H}{\bold{H}}{\bold{F}_A}{\bold{F}_D}{\bold{x}} + {\bold{W}_D^H}{\bold{W}_{\alpha}\big( {\bold{b}} \big)}{\bold{W}_A^H}{\bold{n}}+{\bold{W}_D^H}{\bold{n_q}}.
\end{equation}
where the  dimensions of matrices are 
${\bold{F}_D} \in \mathbb{C}^{N_{rt} \times N_s}$, ${\bold{F}_A} \in \mathbb{C}^{N_t \times N_{rt}}$, ${\bold{H}} \in \mathbb{C}^{N_r \times N_t}$, ${\bold{W}_A^H} \in \mathbb{C}^{N_{rs} \times N_r}$, ${\bold{W}_D^H} \in \mathbb{C}^{N_s \times N_{rs}}$ and $\bold{W}_{\alpha}\big( {\bold{b}} \big) \in \mathbb{R}^{N_{rs} \times N_{rs}}$.

With the diagonal BA matrix $\bold{W}_{\alpha}\big( {\bold{b}} \big)$, we intend to design the precoders  ${\bold{F}_D}$ and ${\bold{F}_A}$, and Combiners ${\bold{W}_D^H}$ and ${\bold{W}_A^H}$, along with the ADC BA ${\bold{W}_{\alpha}\big( {\bold{b}} \big)}$ for a given channel realization $\bold{H}$. We assume perfect CSI at the transmitter. We further assume that 
$N_{rs} = N_s$ and the extension to the case $N_{rs} \ne N_s$ is straightforward.

\section{Precoder Design}\label{PreCoderLb}
The hybrid precoding and combing techniques for systems employing phase shifters in mmWave transceiver architectures impose constraints on them. They require that the entries of the analog precoder $\bold{F}_A$ and combiner ${\bold{W}_A^H}$ have constant magnitude entries. Finding optimal $\bold{F}_A$ and  ${\bold{W}_A^H}$ is quite complex given the number of constraints on their design. Instead, we propose to design the precoder and combiner separately \cite{Paulraj}. The precoder  can be designed using \cite{PreDsgn}. Let the Singular Value Decomposition (SVD) of the channel matrix $\bold{H}$ be $\bold{H} = \bold{U}\bold{\Sigma}\bold{F}_{\text{opt}}^H,{\text{ where }} {\bold{U}} \in \mathbb{C}^{N_r \times N_s}, {\bold{\Sigma}} \in \mathbb{R}^{N_s \times N_s}, {\bold{F}_{\text{opt}}} \in \mathbb{C}^{N_t \times N_s}$.
The hybrid precoders are optimized \cite{PreDsgn} as follows. 
\begin{equation}\label{eq8a}
\begin{aligned}
({\bold{F}_A^{\text{opt}}},{\bold{F}_D^{\text{opt}}}) = & \underbrace{\text{argmin}}_{{\bold{F}_D},{\bold{F}_A}}{\lVert {{\bold{F}_{\text{opt}}} - {{\bold{F}_A}{\bold{F}_D}}} \rVert }_F,\text{ such that }{\bold{F_A}}\in{\mathcal{F}_{RF}}, {\lVert {{\bold{F}_D}{\bold{F}_A}} \rVert }_F^2 = N_s. 
\end{aligned}
\end{equation}
The set $\mathcal{F}_{RF}$ consists of all possible analog precoders that correspond  to phase shifter architecture. This includes all possible $(N_t \times N_{rt})$ matrices with constant magnitude entries.
\section{Bit-allocation Design}\label{ba}
In this section, we derive the optimal BA based on the two criteria (i) and (ii) mentioned above.\\
\vspace{-10mm}
\subsection{Bit-allocation based on MSE minimization criterion}\label{ba_mse}
Having designed the precoders in the previous section such that ${\bold{F}_{\text{opt}}} \approx {\bold{F}_A}{\bold{F_D}}$ with the constraints in ($\ref{eq8a}$), we can rewrite ($\ref{eq5a}$) as
\begin{equation}\label{eq9a}
\begin{aligned}
{\bold{y}} &= {\bold{W}_D^H}{\bold{W}_{\alpha}}{\big( {\bold{b}} \big)}{\bold{W}_A^H}{\bold{U}}{\bold{\Sigma}}{\bold{x}} + {\bold{W}_D^H}{\bold{W}_{\alpha}\big( {\bold{b}} \big)}{\bold{W}_A^H}{\bold{n}}+{\bold{W}_D^H}{\bold{n_q}}.
\end{aligned}
\end{equation}
Using ($\ref{eq9a}$), we derive the expression for $\mbox{MSE}$ $\delta$ as
\begin{equation}\label{10aa}
\begin{gathered}
\delta \triangleq \tr{(E\big[ (\bold{y}-\bold{x})^2\big])}\\
\mbox{MSE}(\bold{x}) = E\big[ (\bold{y}-\bold{x})^2\big] = p({\bold{K}}-{\bold{I}_{N_s}})^2 + {\sigma_{n}^2}{\bold{G}\bold{G}^H} + {\bold{W}_D^H}{\bold{D}_q^2}{\bold{W}_D},
\end{gathered}
\end{equation}
where
${\bold{K}} = {\bold{W}_D^H}{\bold{W}_{\alpha}}{\bold{W}_A^H}{\bold{U}}{\bold{\Sigma}}$, $E[{\bold{x}}{\bold{x}}^H] = p{\bold{I}_{N_s}}$, ${\bold{G}} = {\bold{W}_D^H}{\bold{W}_{\alpha}}{\bold{W}_A^H}$, $E[{\bold{n}}{\bold{n}}^H] = {\sigma_n^2}{\bold{I}_{N_r}}$, $E[{\bold{n_q}}{\bold{n_q}}^H] = {\bold{D}_q^2}$. Note that $p$ is the average power of symbol $\bold{x}$, ${\bold{D}_q^2} = {\bold{W}_{\alpha}}{\bold{W}_{1-\alpha}}{\text{diag}}[ {\bold{W}_A^H}{\bold{H}}({\bold{W}_A^H}{\bold{H}})^H+{\bold{I}_{N_{rs}}}]$, and $E[{\bold{n}}{\bold{n_q}}^H] = 0$. For simplicity of notation, we will refer to $\bold{W}_{\alpha}\big( {\bold{b}} \big)$ as $\bold{W}_{\alpha}$.\\
\indent
We intend to design the combiners $\bold{W}_A^H$ and $\bold{W}_D^H$ and the BA matrix $\bold{W}_{\alpha}$ such that the MSE in ($\ref{10aa}$) is minimized. Thus, we set ${\bold{K}} = {\bold{I}_{N_s}}$, which gives ${\bold{G}\bold{G}^H} = {\bold{\Sigma}^{-2}}$. Hence, the $\mbox{MSE}(\bold{x})$ in ($\ref{10aa}$) is reduced to
\begin{equation}\label{10ac}
\mbox{MSE}(\bold{x}) = {\sigma_n^2}{\bold{\Sigma}^{-2}} + {\bold{W}_D^H}{\bold{D}_q^2}{\bold{W}_D}.
\end{equation}

The first term of $\mbox{MSE}(\bold{x})$ in ($\ref{10ac}$) is channel-dependent and the only design parameter is the second term ${\bold{W}_D^H}{\bold{D}_q^2}{\bold{W}_D}$. Thus, the combiners and BA matrix need to be designed with the condition ${\bold{K}} = {\bold{W}_D^H}\bold{W}_{\alpha}{\bold{W}_A^H}{\bold{U}}{\bold{\Sigma}} = \bold{I}_{N_s},\text{  such that }{\bold{W}_D^H}{\bold{D}_q^2}{\bold{W}_D} = \bold{0}$.\\ 
\indent
This is a hard problem to solve for $\bold{W}_A^H$, $\bold{W}_D^H$, and $\bold{W}_{\alpha}$. Thus, we take a slightly different approach. We show that ($\ref{10ac}$) is indeed the Minimum MSE (MMSE) that can be achieved for a given $\bold{W}_A^H$, $\bold{W}_D^H$, $\bold{W}_{\alpha}$, and channel $\bold{H}$.  This is accomplished by deriving the expression for the CRLB for estimating $\bold{x}$, given the observation $\bold{y}$. For the proof, please refer to Theorem $\ref{Thm1}$ in the Appendix. The expression for the CRLB is derived as
\begin{equation}\label{eq26a}
{\bold{I}^{-1}({\bold{\hat{x}}})} = {\sigma_n^2}{\bold{\Sigma}^{-2}}+{\bold{K}^{-1}}{\bold{W}_D^H}{\bold{D}_q^2}{\bold{W}_D}({\bold{K}^H})^{-1}.
\end{equation}
We minimize the CRLB \cite{MinCRLB} as a function of parameters $\bold{W}_A^H$, $\bold{W}_D^H$, and $\bold{W}_{\alpha}$.
\subsubsection{Minimizing the CRLB}\label{min_crlb}
Given the fact that the MMSE derived using  ($\ref{10ac}$) achieves CRLB for fixed design parameters, we now intend to design the combiners $\bold{W}_A$, $\bold{W}_D$ and the BA matrix ${\bold{W}_{\alpha}}$ by minimizing the CRLB. We wish to have ${\bold{I}^{-1}({\bold{\hat{x}}})}$ in ($\ref{eq26a}$) vanish or gets close to zero.
Substituting  $\bold{K}$  into ($\ref{eq26a}$) and simplifying the equation, we arrive at
\begin{equation}\label{eq29a}
{\bold{I}^{-1}({\bold{\hat{x}}})} = {\sigma_n^2}{\bold{\Sigma}^{-2}}+{\bold{\Sigma}}^{-1}{\bold{U}^H}({\bold{W}_A^H})^{-1}{\bold{W}_{\alpha}^{-1}}{\bold{D}_q^2}{\bold{W}_{\alpha}^{-1}}{\bold{W}_A^{-1}}{\bold{U}}{\bold{\Sigma}}^{-1} \approx \bold{0}.       
\end{equation}

Phase shifters or splitters impose constraints on the design of the analog combiner $\bold{W}_A^H$ \cite{SigProc}. We will denote the constrained analog combiner as $\bold{\tilde{W}}_A^H$. The imperfections in the analog combiner are compensated by the digital combiner, that is ${\bold{W}_A^H} = {\bold{W}_D}{\bold{\tilde{W}}_A^H}$.

We also  would like to design the actual analog combiner ${\bold{\tilde{W}}_A^H}$ and the digital combiner ${\bold{W}_D}$, such that
\begin{equation}\label{eq31a}
{\bold{W}_A^H} = {\bold{U}^H} = {\bold{W}_D}{\bold{\tilde{W}}_A^H}.
\end{equation} 

To design the BA, we substitute ($\ref{eq31a}$) in ($\ref{eq29a}$) to arrive at
\begin{equation}\label{eq32a}
{\bold{I}^{-1}({\bold{\hat{x}}})} = {\bold{\Sigma}^{-2}}\bigg[ {\sigma_n^2}{\bold{I}_{N_s}} +  {\bold{W}_{\alpha}^{-2}}{\bold{D}_q^2}\bigg] \approx \bold{0}.
\end{equation}
\indent
The optimal BA solution can  be posed as a constrained optimization problem shown below.
\begin{equation}\label{bitcond_mse}
\bold{b}^* = \underbrace{\text{argmin}}_{\substack{\bold{b} \in \mathbb{I}^{N_s \times 1}; \\ {P_{\text{TOT}}}\leq{P_{\text{ADC}}}}}\Bigg\{{{\bold{\Sigma}^{-2}}\bigg[ {\sigma_n^2}{\bold{I}_{N_s}} +  {\bold{W}_{\alpha}^{-2}}{\bold{D}_q^2}\bigg]}\Bigg\}
\end{equation}
$P_{\text{TOT}}$ is the total power consumed by the ADCs with bit allocation $\bold{b} = [b_1,b_2,...,b_{N_s}]^T$ and is shown to be $\sum_{i=1}^{N} c{f_s}2^{b_i}$,  where $c$ is the power consumed per conversion step and $f_s$ is the sampling rate in Hz \cite{Uplink}. $P_{\text{ADC}}$ is the allowed ADC power budget.\\ 
\indent
In order to satisfy ($\ref{eq32a}$), the required BA condition becomes
\begin{equation}\label{eq33a}
{\bold{\Sigma}}^{2} \gg {\sigma_n^2}{\bold{I}_{N_s}} + {\bold{W}_{\alpha}^{-2}}{\bold{D}_q^2}.
\end{equation}
Since ${\bold{\Sigma}}^2$, ${\bold{W}_{\alpha}^2}$ and ${\bold{D}_q^2}$ are diagonal matrices, we can rewrite $(\ref{eq33a})$ as a set of $N_s$ inequalities  
\begin{equation}\label{eq34a}
{{\sigma_i}^2} \gg \sigma_n^2 + g(b_i)l_i, \text{  for }1 \le i \le N_s,
\end{equation}
where ${\sigma_i}$ is the diagonal element of ${\bold{\Sigma}}$, ${\sigma_n^2}$ is the noise power, $g(b_i)=\frac{f(b_i)}{1-f(b_i)}$ where $f(b_i)$ \cite{VarBitAlloc} is the ratio of the MQSE and the power of the symbol for a non-uniform MMSE quantizer with $b_i$ bits along the RF path $i$, $i=1,2,...,N_s$. The values for $f(b_i)$ are indicated in the Table \ref{betaVal} and $l_i$ is the $i^{th}$ element of $\text{diag}(\bold{I}_{N_s}+\bold{W}_D^H\bold{\Sigma}^2\bold{W}_D)$.
\begin{table}[http]
\begin{center}
\begin{tabu} to 0.5\textwidth { c c c c c c }
\hline
$b_i$  & 1 & 2 & 3 & 4 & 5 \\
\hline
$f(b_i)$ & 0.3634 & 0.1175 & 0.03454 & 0.009497 & 0.002499 \\
\hline
\end{tabu}
\vspace{1mm}
\caption{Values of $f(b_i)$ for different ADC Quantization Bits $b_i$.} \label{betaVal}
\end{center}
\end{table}
\vspace{-18mm}

We need to satisfy all the $N_s$ inequalities in ($\ref{eq34a}$) to attain the optimal BA. However, it may not always be possible to attain optimal BA, given the number of bits and the power budget. In such scenarios, we would make a best-effort approach to satisfy the set of equations in ($\ref{eq34a}$) and the solution would be the best solution given the power constraint.

For a given $N_s$ and allowable range of ADC bit-resolution $N_b$, we first form a super-set $B_{\text{set}}$ of all possible $\bold{b}_j$'s that satisfy the power budget $P_{\text{ADC}}$.
\begin{equation}\label{eq35a}
B_{\text{set}} \triangleq \Big\{ \bold{b}_j = {\big[ b_{j1}, b_{j2}, \dots, b_{jN}  \big]}^T \text{ for } 0 \leq j < N_b^{N_s} \mid 1 \le b_{ji} \le N_b \text{ and } \sum_{i=1}^{N} cf_s2^{b_{ji}} \leq P_{\text{ADC}} \Big\}
\end{equation}

We incorporate a gain term ${K_{f}}(b_i)$ for a given bit $b_i$ on RF path $i$ into the set of equalities in ($\ref{eq34a}$) such that
\begin{equation}\label{eq36a}
{K_{f}}(b_i) \triangleq \frac{{\sigma_i^2}}{\sigma_n^2+g(b_i)l_i}.
\end{equation}
For a given bit allocation $\bold{b}_j$ in $B_{\text{set}}$, we denote
\begin{equation}\label{eq37a}
{K_{f}}(\bold{b}_j) \triangleq \sum_{i=1}^{N_s} \Bigg[\frac{{\sigma_i^2}}{\sigma_n^2+g(b_i)l_i}\Bigg]. 
\end{equation}
We select ${\bold{b}} \in B_{\text{set}}$ to maximize ${K_{f}}(\bold{b}_j)$ and declare that as the desirable BA solution as
\begin{equation}\label{bitcond_mse}
\bold{b}^* = \underbrace{\text{argmax}}_{\substack{\bold{b}_j \in B_{\text{set}}; \\ {P_{\text{TOT}}}\leq{P_{\text{ADC}}}}} K_f(\bold{b}_j).
\end{equation}
The Algorithm is described in Algorthm~$\ref{AlgoCRLB}$ on page \pageref{AlgoCRLB}.

\subsection{Bit-allocation based on capacity maximization}\label{ba_cap}
In this section, we first derive the expression for the capacity of the SU mmWave Ma-MIMO channel encompassing the channel matrix $\bold{H}$, the hybrid precoders ${\bold{F}_D}$, ${\bold{F}_A}$, and the hybrid combiners ${\bold{W}_D^H}$, ${\bold{W}_A^H}$ along with the BA matrix $\bold{W}_{\alpha}$. We then maximize this capacity expression with respect to the BA matrix for a given power budget to arrive at an optimal BA condition.\\

\subsubsection{Capacity Analysis}\label{cap}
Equation \eqref{eq5a} can be simplified as 
\begin{equation}\label{eq9a_cap}
\begin{aligned}
{\bold{y}} = {\bold{W}_D^H}\bold{W}_{\alpha}{\bold{W}_A^H}{\bold{U}}{\bold{\Sigma}}{\bold{x}} + {\bold{W}_D^H}{\bold{W}_{\alpha}\big( {\bold{b}} \big)}{\bold{W}_A^H}{\bold{n}} + {\bold{W}_D^H}{\bold{n_q}} = {\bold{K}}{\bold{x}} + {\bold{n_1}}
\end{aligned}
\end{equation}
where $\bold{n_1} = {\bold{W}_D^H}{\bold{W}_{\alpha}}{\bold{W}_A^H}{\bold{n}} + {\bold{W}_D^H}{\bold{n_q}}$. Note that $\bold{n}$ and $\bold{n_q}$ are Gaussian random vectors with $\bold{n} \sim \mathcal{CN}(\bold{0},{\sigma_n^2}{\bold{I}_{N_r}})$ and  $\bold{n_q} \sim \mathcal{N}(\bold{0},{\bold{D}_q^2})$, respectively. We also know that $\bold{n_1}$ is $\bold{n_1} \sim \mathcal{N}(\bold{0},\bold{\Phi})$ where $\bold{\Phi} = {\sigma_n^2}{\bold{G}}{\bold{G}^H} + {\bold{W}_D^H}{\bold{D}_q^2}{\bold{W}_D}$ \cite{Zakir2}. 

The instantaneous capacity for a given MIMO channel with ADC power constraint and BA can be written as
\begin{equation}\label{eq14a_cap}
\begin{aligned}
C =  \Bigg\{ \underbrace{\text{max}}_{p(\bold{x}),\bold{b} \in \mathbb{I}^{N_s \times 1},{P_{\text{TOT}}}\leq{P_{\text{ADC}}}}{ I\big(\bold{x}; \bold{y}\big) } \Bigg\}
\end{aligned}
\end{equation}
where $I\big(\bold{x}; \bold{y}\big)$ is the mutual information of random variables $\bold{x}$ and $\bold{y}$.\\
\indent
Note that the mutual information in ($\ref{eq14a_cap}$) is maximized with respect to the BA matrix ${\bold{W}_{\alpha}}{\big( {\bold{b}} \big)}$. Equation \eqref{eq14a_cap} can be written as \cite{Thomas}
\begin{equation}\label{eq15a_cap}
I(\bold{x};\bold{y}) = h(\bold{y}) - h(\bold{y}|\bold{x}) = h(\bold{y}) - h(\bold{K}\bold{x} + \bold{n_1}|\bold{x})
= h(\bold{y}) - h(\bold{n_1})
\end{equation}
We assume that $\bold{x}$ and $\bold{n_1}$ are independent. If $\bold{y} \in \mathbb{C}^{N_s}$, then the differential entropy $h(\bold{y})$ is less than or equal to $\log_2\det(\pi e \bold{Q})$ with equality if and only if $\bold{y}$ is circularly symmetric complex gaussian with $E[\bold{y}\bold{y}^H] = \bold{Q}$ \cite{Bengt}. As such,
\begin{equation}\label{eq16a_cap}
E[\bold{y}\bold{y}^H] = \bold{Q} = E \Big[ (\bold{K}\bold{x} + \bold{n_1})(\bold{K}\bold{x} + \bold{n_1})^H \Big] = E \Big[ \bold{K}\bold{x}\bold{x}^H\bold{K}^H + \bold{n_1}\bold{n_1}^H \Big] = p\bold{K}\bold{K}^H + \bold{\Phi}
\end{equation}
Note that $\bold{\Phi} = {{\sigma_n^2}{\bold{G}}{\bold{G}^H} + {\bold{W}_D^H}{\bold{D}_q^2}{\bold{W}_D}}$. Thus, the differential entropies $h(\bold{y})$ and $h(\bold{n_1})$ satisfy
\begin{equation}\label{diffent}
\begin{split}
h(\bold{y}) &\le \log_2\det(\pi e \bold{Q}) = \log_2\det \bigg( \pi e \Big(p \bold{K}\bold{K}^H + \bold{\Phi} \Big) \bigg), \\
h(\bold{n_1}) &\le \log_2\det(\pi e \bold{\Phi}).
\end{split}
\end{equation}
The Theorem~$\ref{Thm2}$ in the Appendix proves that $\bold{n_1}$ is a circularly symmetric jointly Complex Gaussian vector. Hence, we can write
\begin{equation}\label{hn_equal}
h(\bold{n1}) = \log_2\det(\pi e \bold{\Phi}).
\end{equation}
Thus, the maximum mutual information $I(\bold{X};\bold{Y})$ achieved can be written as
\begin{equation}\label{maxI}
I(\bold{X};\bold{Y}) \overset{(a)}=  h(\bold{y}) - h(\bold{n_1}) = \log_2\det( \pi e \bold{Q}) - \log_2\det(\pi e \bold{\Phi})
= \log_2\det \Big ( p\bold{K}\bold{K}^H\bold{\Phi}^{-1} + \bold{I}_{N_s} \Big)
\end{equation}
where (a) follows from the assumption that the input symbol vector $\bold{x}$ is circular symmetric Gaussian vector that could be modeled 
as $\bold{x} \sim \mathcal{CN}(\bold{0},p\bold{I_{N_s}})$. We simplify ($\ref{maxI}$) as
\begin{equation}\label{maxI_cont}
\begin{split}
I(\bold{X};\bold{Y}) &= \log_2\det \Big ( p\bold{K}\bold{K}^H\bold{\Phi}^{-1}\bold{K}\bold{K}^{-1} + \bold{K}\bold{K}^{-1} \Big) = \log_2\det \Big ( p\bold{K} \big ( \bold{K}^H\bold{\Phi}^{-1}\bold{K} + \frac{1}{p}\bold{I}_{N_s} \big) \bold{K}^{-1} \Big)\\
&= \log_2\det ( p\bold{K}) \det \Big ( \bold{K}^H\bold{\Phi}^{-1}\bold{K} + \frac{1}{p}\bold{I}_{N_s} \Big) \det (\bold{K}^{-1}) = \log_2 p^{N_s}\det \Big ( \bold{K}^H\bold{\Phi}^{-1}\bold{K} + \frac{1}{p}\bold{I}_{N_s} \Big).
\end{split}
\end{equation}
The capacity is computed by maximizing $I(\bold{X};\bold{Y})$ for a given channel $\bold{H}$, and for a given combiner pair $\bold{\tilde{W}}_A^H$ and $\bold{W}_D^H$. Hence, the maximization of \eqref{maxI_cont} will be over the BA matrix $\bold{W}_{\alpha}$. Note that $\bold{\Phi}$ is a function of $\bold{W}_{\alpha}$ and $\bold{K}$. Thus
\begin{equation}\label{maxI_cont2}
C = \max \Big\{ \log_2 p^{N_s}\det \Big ( \bold{K}^H\bold{\Phi}^{-1}\bold{K} + \frac{1}{p}\bold{I}_{N_s} \Big) \Big\} = {N_s}\log_2 p + \log_2\det \Big ( ({\bold{I}^{-1}({\bold{\hat{x}}})})^{-1} + \frac{1}{p}\bold{I}_{N_s} \Big).
\end{equation}
Note that ${\bold{I}^{-1}({\bold{\hat{x}}})}$ is the CRLB in \eqref{eq24a} achieved by the MSE $\delta$ \eqref{10aa}.
%
\subsubsection{Maximizing the capacity for optimal bit allocation}\label{cap_max}
Using the capacity expression derived in ($\ref{maxI_cont2}$), the capacity is maximized by selecting some $\bold{b}^*$  for an optimal BA  that satisfies the ADC power constraint. We can write this expression for the maximum capacity from ($\ref{maxI_cont2}$) as
\begin{equation}\label{maxcap}
C = {N_s}\log_2p \text{ }+\underbrace{\text{max}}_{\substack{\bold{b}^*,{P_{\text{TOT}}}\leq{P_{\text{ADC}}}}}\Bigg\{ \log_2  \det \Big ( ({\bold{I}^{-1}({\bold{\hat{x}}})})^{-1} + \frac{1}{p}\bold{I}_{N_s} \Big) \Bigg\}.
\end{equation}
The condition for $\bold{b}^*$ that optimizes ($\ref{maxcap}$) is given by
\begin{equation}\label{bitcond}
\bold{b}^* = \underbrace{\text{argmax}}_{\substack{\bold{b} \in \mathbb{I}^{N_s \times 1}, \\ {P_{\text{TOT}}}\leq{P_{\text{ADC}}}}}\Bigg\{ \log_2  \det \Big ( ({\bold{I}^{-1}({\bold{\hat{x}}})})^{-1} + \frac{1}{p}\bold{I}_{N_s} \Big) \Bigg\}.
\end{equation}
By substituting $\bold{K}$ into ($\ref{eq26a}$) and by designing the structure of the hybrid combiner as in $(\ref{eq31a})$, we can simplify the expression for CRLB as
\begin{equation}\label{crlb_cont}
{\bold{I}^{-1}({\bold{\hat{x}}})} = {\sigma_n^2}{\bold{\Sigma}^{-2}}+{\bold{\Sigma}}^{-1}{\bold{U}^H}({\bold{W}_A^H})^{-1}{\bold{W}_{\alpha}^{-1}}{\bold{D}_q^2}{\bold{W}_{\alpha}^{-1}}{\bold{W}_A^{-1}}{\bold{U}}{\bold{\Sigma}}^{-1} = {\sigma_n^2}{\bold{\Sigma}^{-2}}+{\bold{\Sigma}}^{-2}{\bold{W}_{\alpha}^{-2}}{\bold{D}_q^2}
\end{equation}
We now compute the Inverse of CRLB $\Big({\bold{I}^{-1}({\bold{\hat{x}}})}\Big)^{-1}$ as 
\begin{equation}\label{inv_crlb}
\Big({\bold{I}^{-1}({\bold{\hat{x}}})}\Big)^{-1} = \Big({\sigma_n^2}{\bold{\Sigma}^{-2}}+{\bold{\Sigma}}^{-2}{\bold{W}_{\alpha}^{-2}}{\bold{D}_q^2}\Big)^{-1} = \text{diag}\bigg( \frac{\sigma_1^2}{\sigma_n^2 + g(b_1)l_1}, \cdots, \frac{\sigma_{N_s}^2}{\sigma_n^2 + g(b_{N_s})l_{N_s}}\bigg),
\end{equation}
Substituting $\Big({\bold{I}^{-1}({\bold{\hat{x}}})}\Big)^{-1}$ evaluated in ($\ref{inv_crlb}$) in ($\ref{bitcond}$), we have
\begin{equation}\label{bitcond_cont}
\begin{split}
\bold{b}^* &= \underbrace{\text{argmax}}_{\substack{\bold{b} \in \mathbb{I}^{N_s \times 1}, \\ {P_{\text{TOT}}}\leq{P_{\text{ADC}}}}} \log_2 \det \text{diag} \Big( \frac{\sigma_1^2}{\sigma_n^2 + g(b_1)l_1} + \frac{1}{p},\cdots, \frac{\sigma_{N_s}^2}{\sigma_n^2 + g(b_{N_s})l_{N_s}} + \frac{1}{p} \Big),  \\
&= \underbrace{\text{argmax}}_{\substack{\bold{b} \in \mathbb{I}^{N_s \times 1}, \\ {P_{\text{TOT}}}\leq{P_{\text{ADC}}}}} \log_2 \prod_{i=1}^{N_s} \Bigg( \frac{\sigma_i^2}{\sigma_n^2 + g(b_i)l_i} + \frac{1}{p} \Bigg) = \underbrace{\text{argmax}}_{\substack{\bold{b} \in \mathbb{I}^{N_s \times 1}, \\ {P_{\text{TOT}}}\leq{P_{\text{ADC}}}}} \sum_{i=1}^{N_s} \bigg\{ \log_2  \Big( q(b_i) + 1 \Big) \bigg\},
\end{split}
\end{equation}
where $q(b_i) = \frac{p\sigma_i^2}{\sigma_n^2 + g(b_i)l_i}$.
The term $\log_2  \Big( q(b_i) + 1 \Big)$ can be expanded for two scenarios given below.\\
\textit{Case 1:}
For the case of $0 \leq q(b_i) < 1$,  we have $\log_2  \Big( q(b_i) + 1 \Big) \simeq \frac{q(b_i)}{\ln2}$ with proof provided  in the Appendix (Lemma $\ref{lemm1}$).
Thus, the maximization in $\eqref{bitcond_cont}$ can be written as
\begin{equation}\label{bitcond_cont2}
\bold{b}^* = \underbrace{\text{argmax}}_{\substack{\bold{b} \in \mathbb{I}^{N_s \times 1}, \\ {P_{\text{TOT}}}\leq{P_{\text{ADC}}}}} \sum_{i=1}^{N_s} \frac{p\sigma_i^2}{\sigma_n^2 + g(b_i)l_i}.
\end{equation}\\
\textit{Case 2:} 
For the case $1 \leq q(b_i) < \infty$, we can use Taylor series  to arrive at 
$\log_2  \Big( q(b_i) + 1 \Big) = \Bigg(1-\frac{1}{q(b_i)}\Bigg)P + L(p,\sigma_i^2, \sigma_n^2)$. Proof is provided in Lemma $\ref{lemm2}$ in the Appendix.
$P$ and $L(p,\sigma_i^2, \sigma_n^2)$ are independent of $b_i$. Hence, the maximization in $\eqref{bitcond_cont}$ can be simplified to
\begin{equation}\label{bitcond_cont1}
\bold{b}^* = \underbrace{\text{argmax}}_{\substack{\bold{b} \in \mathbb{I}^{N_s \times 1}, \\ {P_{\text{TOT}}}\leq{P_{\text{ADC}}}}} \sum_{i=1}^{N_s} \Bigg(1-\frac{1}{q(b_i)}\Bigg) = \underbrace{\text{argmax}}_{\substack{\bold{b} \in \mathbb{I}^{N_s \times 1}, \\ {P_{\text{TOT}}}\leq{P_{\text{ADC}}}}} \sum_{i=1}^{N_s} \frac{p\sigma_i^2}{\sigma_n^2 + g(b_i)l_i}.
\end{equation}
We observe that $\eqref{bitcond_cont1}$ coincides with $\eqref{bitcond_cont2}$. Hence both scenarios lead to the same optimization problem. We now define the term $K_f(b_i)$ for a given $b_i$ on RF path $i$ as
\begin{equation}\label{gain_term}
k_f(b_i) \triangleq \frac{p\sigma_i^2}{\sigma_n^2 + g(b_i)l_i}.
\end{equation}
For a given $\bold{b}_j \in B_{\text{set}}$ defined in \eqref{eq35a}, we define
\begin{equation}\label{gain_termsum}
K_f(\bold{b}_j) \triangleq \sum_{i=1}^{N_s} \frac{p\sigma_i^2}{\sigma_n^2 + g(b_i)l_i}.
\end{equation}
The maximization in $\eqref{bitcond_cont1}$ can than be written as
\begin{equation}\label{bitcond_cont3}
\bold{b}^* = \underbrace{\text{argmax}}_{\substack{\bold{b}_j \in B_{\text{set}}, \\ {P_{\text{TOT}}}\leq{P_{\text{ADC}}}}} K_f(\bold{b}_j).
\end{equation}
\indent
Interestingly, this is the similar to MSE criterion in $(\ref{bitcond_mse})$ to obtain the optimal BA.\\
\indent
For a high-resolution ADC, we have $\bold{D}_q^2 = \bold{0}$ and the CRLB defined in ($\ref{crlb_cont}$) reduces to ${\bold{I}^{-1}({\bold{\hat{x}}})} = {\sigma_n^2}{\bold{\Sigma}^{-2}}$. Hence, the capacity with infinite-resolution ADC's $C_{\infty}$ can be derived as
\begin{equation}\label{infcap}
C_{\infty} = \log_2 p^{N_s}\det \Big ( \frac{1}{\sigma_n^2}\bold{\Sigma}^2 + \frac{1}{p}\bold{I}_{N_s} \Big) = \log_2 \det \Big ( \frac{p}{\sigma_n^2}\bold{\Sigma}^2 + \bold{I}_{N_s} \Big).
\end{equation}
\vspace{-0.2in}
With uniform power allocation on the transmitter, $p$ is uniformly divided amoung $N_s$ RF paths and the capacity with uniform power allocation at the transmitter becomes
\begin{equation}\label{infcap_unfm}
C_{\infty} = \sum_{i=1}^{N_s} {\log_2 \Bigg( \frac{\rho}{N_s}\sigma_i^2 + 1\Bigg)},
\end{equation}
where $\rho = \frac{p}{\sigma_n^2}$ is the average SNR at the receiver.\\
Similarly, with perfect CSI at the transmitter and waterfilling, the capacity with high-resolution ADCs can be written as
\begin{equation}\label{infcap_waterfill}
C_{\infty} = \sum_{i=1}^{N_s} {\log_2 \Bigg( \epsilon_i\frac{\rho}{N_s}\sigma_i^2 + 1\Bigg)},
\end{equation}
where $\epsilon_i$ is the portion of the total power $p$ allocated to RF path $i$ at the transmitter based on water-filling algorithm \cite{Vishwa}. Thus ($\ref{infcap_unfm}$) and ($\ref{infcap_waterfill}$) are special cases of \eqref{maxcap}.
\subsubsection{CRLB-based Bit Allocation Algorithm}\label{Algo}
An algorithm to compute the BA based on the MSE minimization $(\ref{bitcond_mse})$ or capacity maximization $(\ref{bitcond_cont3})$ is provided in Algorthm~$\ref{AlgoCRLB}$.
\begin{algorithm}
  \caption{CRLB-based Bit Allocation}\label{AlgoCRLB}
  \begin{algorithmic}
    \Procedure{CRLB-based Bit Allocation}{$\bold{H},\bold{W}_A,\bold{\Sigma},B_{\text{set}},N_s, g(\cdot),\sigma_n^2$}
      \State $\bold{H}\gets \text{MIMO channel}$
      \State $\bold{W}_A\gets \text{Combiners designed as per ($\ref{eq31a}$)}$
      \State $\bold{\Sigma} \gets \text{Matrix containing singular values $\sigma_i$}$
      \State $B_{\text{set}}\gets \text{Bit allocations adhering to ADC power budget}$
      \State $N_s\gets \text{Number of spatially-multiplexed paths}$
      \State $g(\cdot)\gets \text{Quantization error lookup table}$
      \State $\sigma_n^2\gets \text{AWGN power}$
      \For{\texttt{j=0;j++ ;until j<size of $B_{\text{set}}$}}
           \State $K_f(b_j) = 0$
           \For{\texttt{i=0;i++ ;until i<$N_s$}}
                 \State \footnotesize{$K_f(b_j)\gets K_f(b_j) + \frac{{\sigma_i^2}}{\sigma_n^2+g(b_i)l_i}$}
                 \normalsize
           \EndFor
     \EndFor
     \State $index\gets \text{max}(K_f)$
     \State $\bold{b}\gets B_{\text{set}}\text{ at } index$
     \State \textbf{return} $\bold{b}$ \Comment{Optimal Bit Allocation Vector}
    \EndProcedure
  \end{algorithmic}
\end{algorithm}
\vspace{-5mm}
\section{Combiner Design}\label{comb_dsgn}
We design the hybrid combiner similar to hybrid precoder using the optimal structure  derived in $(\ref{eq31a})$. This requires that the left singular matrix $\bold{U}$ of the channel $\bold{H}$ is factored as the product of the constrained analog combiner $\bold{\tilde{W}}_A^H$ and the digital combiner ${\bold{W}_D}$. The hybrid combiner is derived by solving the optimization problem using  method described in \cite{PreDsgn}. 
\begin{equation}\label{comb_desgn}
({\bold{\tilde{W}}_A^{opt}},{\bold{W}_D^{opt}}) =  \underbrace{\text{argmin}}_{{\bold{\tilde{W}}_A},{\bold{W}_D}}{\lVert {{\bold{U}_{\text{opt}}} - {{\bold{\tilde{W}_A}}{\bold{W}_D^H}}} \rVert }_F,\text{ such that }{\bold{\tilde{W}_A}}\in{\mathcal{W}_{RF}}, {\lVert {{\bold{W}_D^H}{\bold{\tilde{W}}_A}} \rVert }_F^2 = N_s 
\end{equation}
$\mathcal{W}_{RF}$ is the set of all possible analog combiners architecture based on phase shifters. This includes all possible $N_r \times N_s$ matrices with constant magnitude entries.
\section{Simulations and Numerical Results} \label{Sim}
We simulate the mmWave channel using the NYUSIM channel simulator with 2 dominant scatters, the configurations specified in Table $\ref{nyusimtab}$ $\cite{nyusim}$. We consider $N_s=8 \text{ or } 12$ strong channels for MSE evaluations and capacity simulations. The combiners are designed as per ($\ref{eq31a}$).\\
\indent
With the above channel, we run the simulations to evaluate the MSE $\delta$ derived in $(\ref{10aa})$ and the capacity as derived in ($\ref{maxI_cont2}$). The plots indicating the MSE $\delta$ at various SNRs are shown for $N_s = 8$ and $N_s = 12$ in Figures $\ref{fig:crlb_Nr8_H64by32.eps}$ and  $\ref{fig:crlb_Nr12_H64by32.eps}$, respectively. Capacity simulations at various SNRs for $N_s = 8$ and $N_s = 12$ are shown in Figures $\ref{fig:crlb_Capacity_Nr8_H64by32.eps}$ and 
 $\ref{fig:crlb_Capacity_Nr12_H64by32.eps}$, respectively. The MSE $\delta$ and the capacity simulations are evaluated with1-bit ADCs, 2-Bit ADCs and with no quantization across all RF paths. This is indicated in Figures $\ref{fig:crlb_Nr8_H64by32.eps}$, $\ref{fig:crlb_Nr12_H64by32.eps}$, $\ref{fig:crlb_Capacity_Nr8_H64by32.eps}$ and $\ref{fig:crlb_Capacity_Nr12_H64by32.eps}$ using lines (a), (b) and (d) respectively. MSE $\delta$ $(\ref{10aa})$ and the capacity ($\ref{maxI_cont2}$) obtained for various SNRs using the ES BA are indicated using the line (c) in the Figures $\ref{fig:crlb_Nr8_H64by32.eps}$, $\ref{fig:crlb_Nr12_H64by32.eps}$, $\ref{fig:crlb_Capacity_Nr8_H64by32.eps}$, and $\ref{fig:crlb_Capacity_Nr12_H64by32.eps}$. We arrive at the BA solution based on our proposed approach of minimizing the MSE (maximizing capacity) in ($\ref{bitcond_cont3}$) and we see that the MSE $\delta$ and the capacity obtained with the BA solution indicated by (line e) is very close to the ES BA solution (line c).\\
\indent
In addition we evaluate the MSE $\delta$ for the MQSE BA using revised minimization of MQSE (rev-MMQSE) as defined in \cite{VarBitAllocJour,VarBitEner}. We employ the SVD based hybrid combiner with this BA scheme. The MSE $\delta$ at various SNRs for $N_s = 8$ and $N_s = 12$ are indicated in Figures $\ref{fig:crlb_Nr8_H64by32.eps}$ and $\ref{fig:crlb_Nr12_H64by32.eps}$ respectively using lines (f). Similarly the capacity evaluated with rev-MQSE BA is plotted using lines (f) in the figures $\ref{fig:crlb_Capacity_Nr8_H64by32.eps}$ and $\ref{fig:crlb_Capacity_Nr12_H64by32.eps}$ for $N_s = 8 \text{ and }12$, respectively.\\
\begin{figure}[t]
\centering
\includegraphics[width=0.5\textwidth]{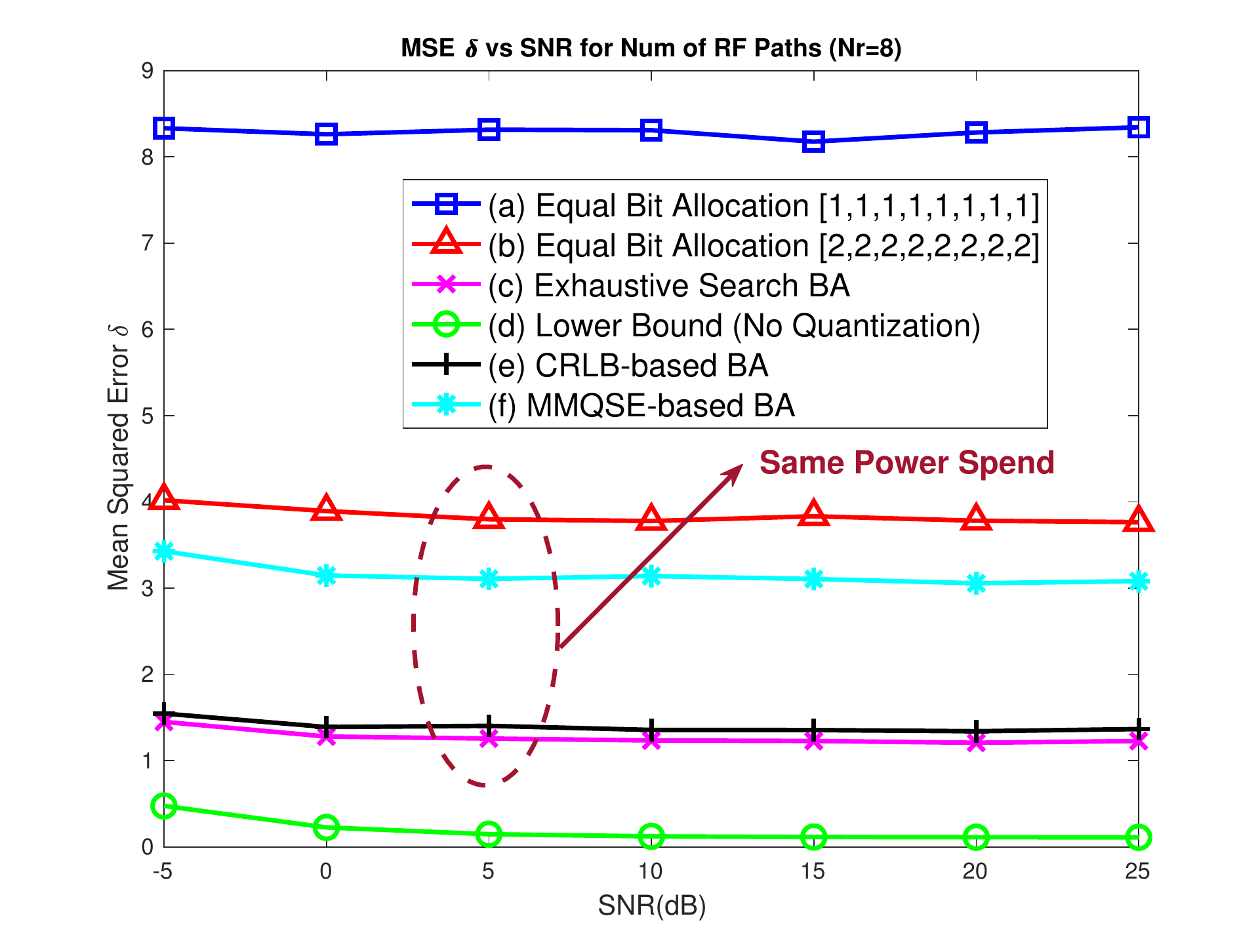}
\caption{MSE $\delta$ vs. SNR for $N_s=8$ for all 1-bit, 2-bit, ES BA, MSQE BA and CRLB BA.}
\label{fig:crlb_Nr8_H64by32.eps}
\includegraphics[width=0.5\textwidth]{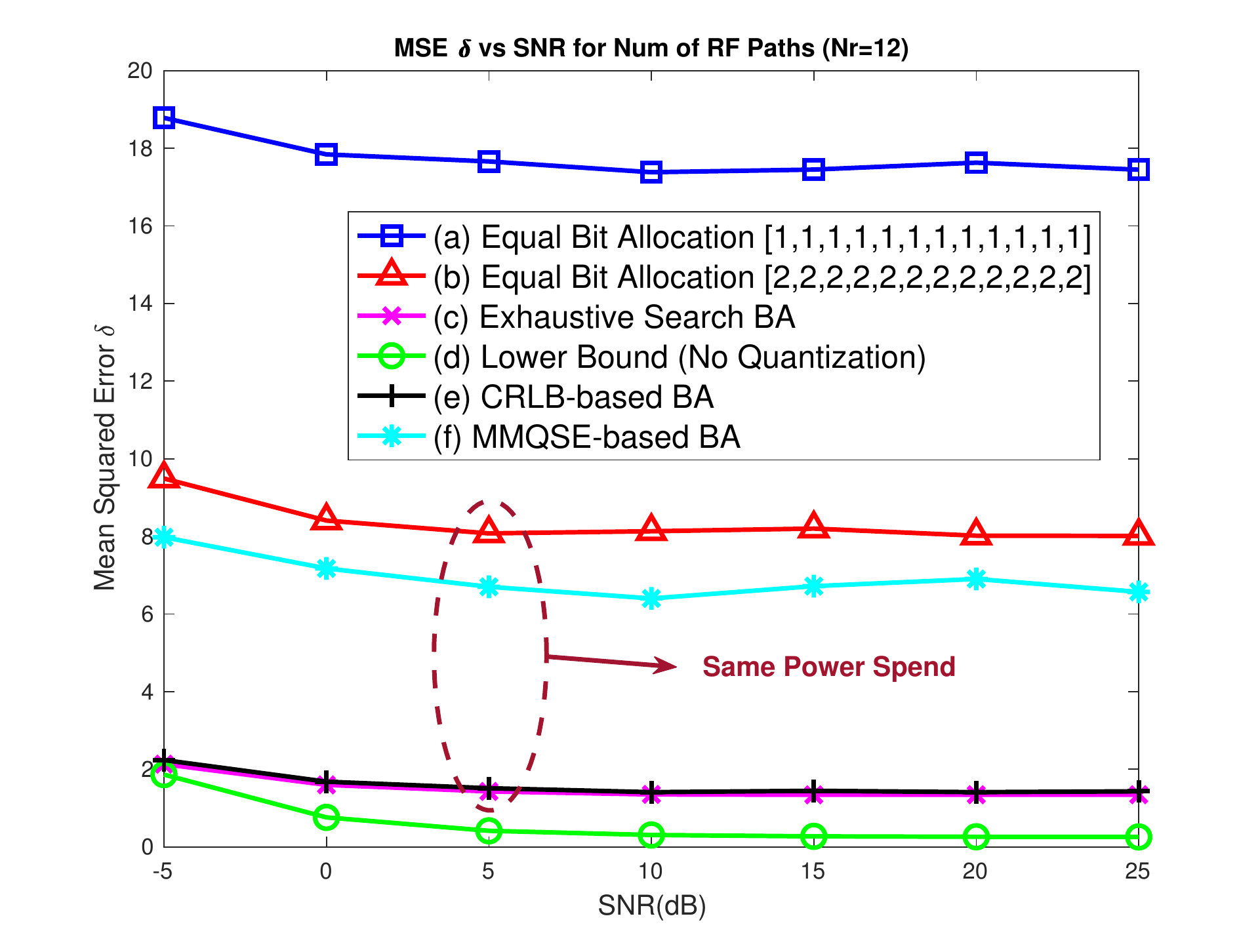}
\caption{MSE $\delta$ vs. SNR for $N_s=12$ for all 1-bit, 2-bit, ES BA, MSQE BA and CRLB BA.}
\label{fig:crlb_Nr12_H64by32.eps}
\end{figure}
\vspace{-0.3in}
\begin{figure}[tp]
\centering
\includegraphics[width=0.5\textwidth]{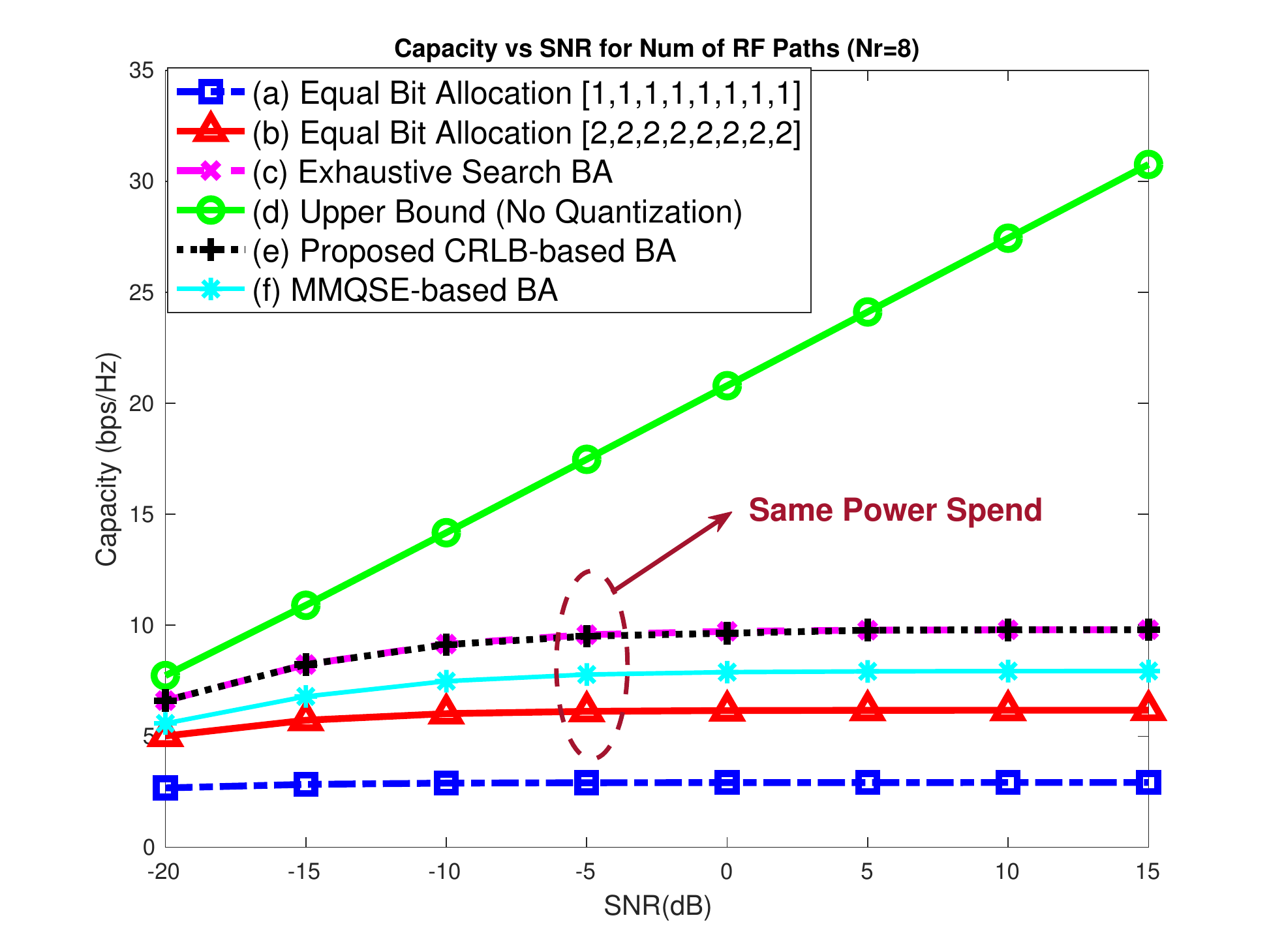}
\caption{Capacity vs. SNR for $N_s=8$ for all 1-bit, 2-bit, ES BA, MSQE BA and CRLB BA.}
\label{fig:crlb_Capacity_Nr8_H64by32.eps}
\includegraphics[width=0.5\textwidth]{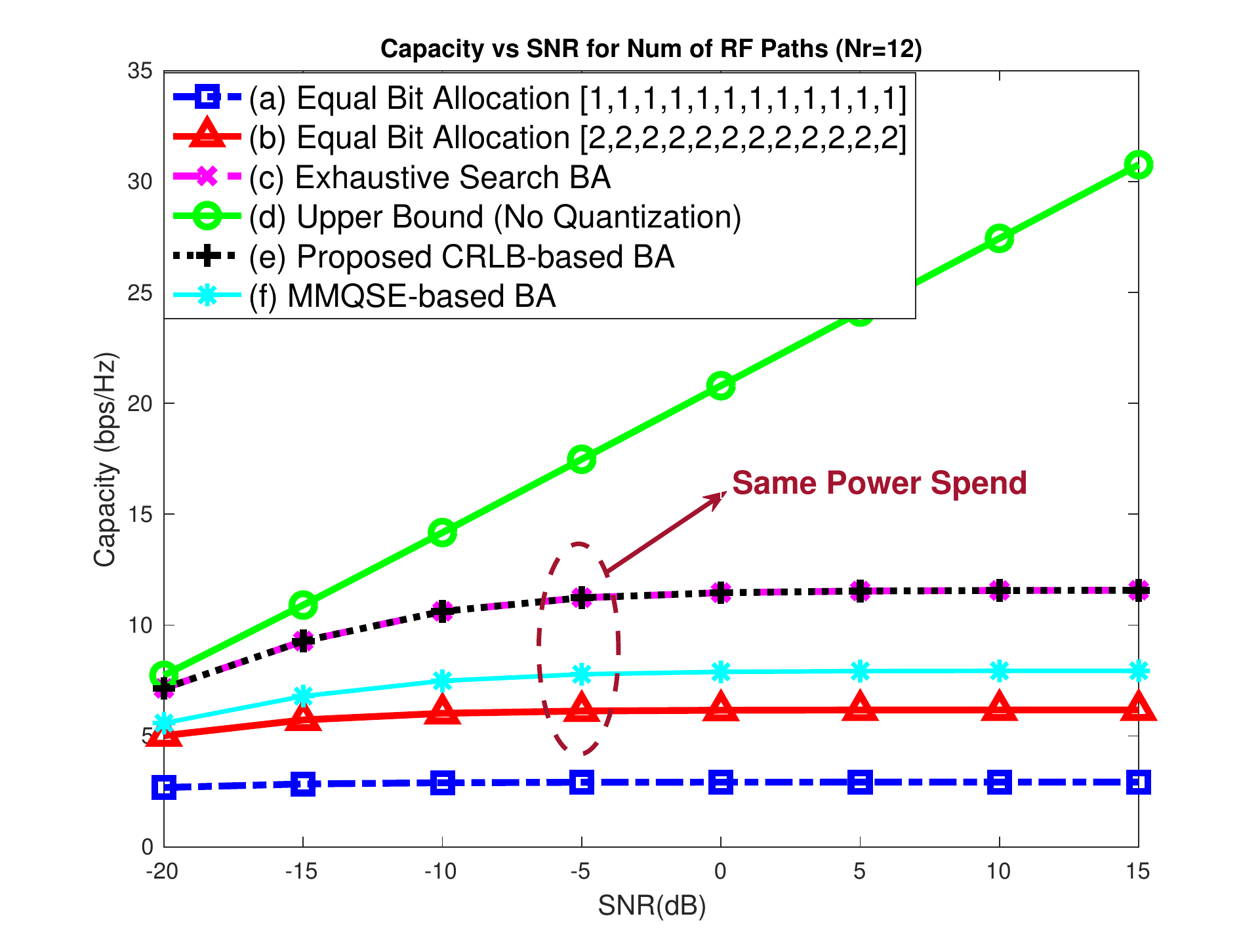}
\caption{Capacity vs. SNR for $N_s=12$ for all 1-bit, 2-bit, ES BA, MSQE BA and CRLB BAs.}
\label{fig:crlb_Capacity_Nr12_H64by32.eps}
\end{figure}
\vspace{-0.3in}
\begin{table}
\begin{center}
\begin{tabu} to 0.5\textwidth {| l| l| }
 \hline
 \textbf{Parameters}  & \textbf{Value/Type} \\
 \hline
Frequency & 28Ghz \\
\hline
Environment & Line of sight \\
\hline
T-R seperation & 100m\\
\hline
TX/RX array type & ULA\\
\hline 
Num of TX/RX elements $N_t$/$N_r$ & 32/64\\
\hline
TX/RX  antenna spacing & $\lambda/2$\\
\hline
\end{tabu}
\vspace{1mm}
\caption{$\text{Channel parameters for NYUSIM model $\cite{nyusim}$}$} \label{nyusimtab}.
\vspace{-1in}
\end{center}
\end{table}
\subsection{Computational Complexity}\label{speed}
We evaluate the computational complexity in terms of the number of multiplications and additions required to arrive at the BA. We assume that the Analog combiner ${\bold{\tilde{W}}_A^H}$ and the digital combiner $\bold{W}_D^H$ are derived as defined in the previous section. We analyze the computational complexity of ES, MQSE and proposed CRLB based BA algorithms.\\
\indent
It can be seen that ES BA requires $\gamma \big(N_s^2 + 2N_s\big)$ complex multiplications, $3N_s^2$ real multiplications and $\gamma \big(N_s(N_s-1) + N_s\big)$ complex additions. Here $\gamma$ is the number of MSE $\delta$ evaluations.\\
\indent
In case of the proposed BA, we precompute the gain term $K_f(b_i)$ in $(\ref{gain_term})$ for all allowable bits (e.g, 1 to 4) and for all $N_s$ RF paths. The computation of $l_i$'s require computation of ${\text{diag}}[ {\bold{W}_D^H}{\bold{\Sigma}}^2{\bold{W}_D}+{\bold{I}_{N_s}}]$. This requires no complex multiplications or additions, and $3N_s^2$ real multiplications and $2N_s^2 + N_s(N_s-1)$ real additions. For the computation of ${K_{f}}(\bold{b}_j)$ as defined in $(\ref{gain_termsum})$ for all possible BA's that satisfy the ADC power constraint, we require $\mu \big(N_s-1 \big)$ real additions. Thus a total of $3N_s^2 + 3N_sN_b$ real multiplications and $3N_s^2 + N_sN_b + \mu \big(N_s-1 \big)$ real additions with our proposed BA. Here, $N_b$ is the number of ADC bits resolution range and $\mu$ is the number of evaluations of ${K_{f}}(\bold{b}_j)$.\\
\indent
For the rev-MMQSE BA algorithm proposed in \cite{VarBitEner}, the optimal binary search is done by precomputing and storing the term $\log_2\Bigg( \frac{\left\vert\left\vert\big[ \bold{H}_b \big]_{i,:}\right\vert\right\vert^{\frac{2}{3}}}{\sum_{j=1}^{N_s}\left\vert\left\vert\big[ \bold{H}_b \big]_{j,:}\right\vert\right\vert^{\frac{2}{3}}}\Bigg)$ for $1 \leq i \leq N_s$ as defined in proposition 1 of \cite{VarBitEner}. This requires $3N_s^2$ real multiplications and $N_s^2+N_s(N_s-1)$ real additions. A $T^{th}$ order polynomial evaluation of the cube root and $\log_2$ is assumed for the above term. We assume that the BA solution is arrived at half stage of the binary search, with number of allowable bits between 1-4 on each RF path. The binary search would only require real additions as the $\log_2$ terms have been precomputed and stored for $1 \leq i \leq N_s$. As a result, the rev-MMQSE BA on average requires $N_s(3N_s+T^2+T+1)$ real multiplications and $2N_s^2+N_s(2T-1)+3(N_s-1)\log_2N_s$ real additions.
\begin{table}
\begin{center}
\begin{tabu} to 1.0\textwidth { | X[c] | X[c] | X[c] | X[c] | X[c] | X[c] | X[c] |}
\hline
\multirow{1}{1.5cm}{$N_s$} & \multicolumn{3}{c|}{Number of complex multiplications} & %
    \multicolumn{3}{c|}{Number of complex additions}\\
\cline{2-7}
& \centering \small ES & \small MQSE-based & \centering \small CRLB-based & \centering \small ES & \small MQSE-based & \centering \small CRLB-based\\
\hline
\centering \multirow{2}{*}{8}  & \small1502400 & \small 440$^{\S}$ & \small \textcolor{red}{288$^{\S}$} & \small 1201920 & \small 263$^{\dagger}$ & {\small \textcolor{red}{13370$^{\dagger}$}} \\
& {\small {192$^{\S}$}} & & & & &\\
\hline
\centering \multirow{2}{*}{12} & \small 223865040 & \small 804$^{\S}$ & \small \textcolor{red}{576$^{\S}$} & \small191884320 & \small 528$^{\dagger}$ & {\scriptsize \textcolor{red}{1466263$^{\dagger}$}}\\
& {\small {432$^{\S}$}} & & & & &\\
\hline
\end{tabu}
\footnotesize{$^{\S}$ Real multiplications. $^{\dagger}$ Real additions} \\
\caption{Computational complexity in terms of total number of multiplications and additions ($N_b$=4 and $T=5$)} \label{tab:CRLBTab1}
\end{center}
\end{table}
\vspace{-0.7in}
\section{Conclusion}\label{conc}
The SU mmWave Ma-MIMO is a typical use case in the deployment of the back-haul wireless links between the Base Stations (BS) considered in the 5G standards. In this paper, we study the adoption of variable-resolution ADCs in the SU mmWave Ma-MIMO receivers. Using the SVD of the Ma-MIMO channel matrix we arrived at an optimal bit-allocation (BA) condition based on the criterion of (i) minimizing MSE of the received, quantized and combined signal and (ii) maximizing the capacity of the Ma-MIMO channel encompassing hybrid precoder and hybrid combiner, under a receiver power constraint. Both these criteria lead to the same BA as a consequence of the relationship between the capacity and the CRLB that we derived. In case of (i) we showed that the MSE approaches the CRLB, with the CRLB being a function of hybrid precoder, combiner and BA matrix. We minimize the CRLB with respect to BA matrix as the MSE minimization poses multiple constraints. In the case of (ii) we derive the expression for the capacity of the Ma-MIMO and show that it is a function of the CRLB in (i). We compared the performance and computational complexity of the proposed BA techniques with ES BA, MQSE BA. It is seen that minimizing the MQSE under power constraint doesn't always ensure optimal MSE or capacity performance. We show that the MSE and capacity performance of the proposed BA is very close to that of the ES BA. The computational complexity of our proposed method has significant improvement compared to ES BA method and slightly inferior to MQSE based BA. The  increase in computational expense comes at a significant improvement in MSE and capacity performance.
\appendix
\section{Appendix}\label{FirstAppendix}
\renewcommand{\thesubsection}{\Alph{subsection}}
\numberwithin{equation}{section}
\newtheorem{theorem}{Theorem}
\begin{theorem}\label{Thm1}
If $\mbox{MSE}(\bold{x})$ is the Mean Square Error matrix as defined in ($\ref{10ac}$), and if ${\bold{\hat{x}}}$ is the estimate of $\bold{x}$ given the observation $\bold{y}$ in ($\ref{eq9a}$) for a given fixed $\bold{H}$, $\bold{W}_A^H$, $\bold{W}_D^H$, ${\bold{W}_{\alpha}}{\big( {\bold{b}} \big)}$, then there exists an estimator that is efficient. That is, $\mbox{MSE}(\bold{x})$ achieves the CRLB ${\bold{I}^{-1}({\bold{\hat{x}}})}$ under similar conditions, in other words $\mbox{MSE}(\bold{x})$ is indeed MMSE.
\end{theorem}
\begin{proof}
We look at the problem in ($\ref{eq9a}$) as an Estimation problem, such that we need to estimate $\bold{x}$ given $\bold{y}$ is observed; given that $\bold{W}_A^H$, $\bold{W}_D^H$, ${\bold{W}_{\alpha}}$ are fixed. ($\ref{eq9a}$) can be rewriten  as 
\begin{equation}\label{eq13a}
{\bold{y}} = {\bold{K}}{\bold{x}} + {\bold{n_1}}
\end{equation}
where ${\bold{K}} = {\bold{W}_D^H}{\bold{W}_{\alpha}}{\bold{W}_A^H}{\bold{U}}{\bold{\Sigma}}$, and 
$\bold{n_1} = {\bold{W}_D^H}{\bold{W}_{\alpha}}{\bold{W}_A^H}{\bold{n}} + {\bold{W}_D^H}{\bold{n_q}}.$
We know that $\bold{n}$ and $\bold{n_q}$ are Gaussian random vectors, with the following statistics:
\vspace{-12mm}
\begin{center}
\begin{equation}\label{eq16a}
\bold{n} \sim \mathcal{N}(\bold{0},{\sigma_n^2}{\bold{I}_{N_r}}),\text{  }\bold{n_q} \sim \mathcal{N}(\bold{0},{\bold{D}_q^2}).
\end{equation}
\end{center}
The statistical distribution of $\bold{n_1}$ is given by: 
\begin{equation}\label{eq17a}
E[\bold{n_1}] = {\bold{W}_D^H}{\bold{W}_{\alpha}}{\bold{W}_A^H}E[{\bold{n}}] + {\bold{W}_D^H}E[{\bold{n_q}}] = {\bold{0}},
\end{equation}
\begin{equation}\label{eq18a}
{\sigma_{n_1}^2} = E[ (\bold{n_1} - E[\bold{n_1}])^2 ] = E[{\bold{n_1}}{\bold{n_1}}^H] = {\sigma_n^2}{\bold{G}}{\bold{G}^H} + {\bold{W}_D^H}{\bold{D}_q^2}{\bold{W}_D}.
\end{equation}
%
%
Thus, the statistics of $\bold{n_1}$ follows: 
\begin{equation}\label{eq22a}
\bold{n_1} \sim \mathcal{N}(\bold{0},({{\sigma_n^2}{\bold{G}}{\bold{G}^H} + {\bold{W}_D^H}{\bold{D}_q^2}{\bold{W}_D}})).
\end{equation}
%
Equation ($\ref{eq13a}$) can be seen as a linear model, in which we intend to estimate $\bold{x}$, given the observation $\bold{y}$. We can express the conditional probability distribution of ${\bold{y}}$ given ${\bold{x}}$ as $\cite{Kay}$
\begin{equation}\label{eq23a}
p({\bold{y} \vert {\bold{x}}}) \sim \frac{1}{({2\pi}{{\sigma_{n_1}^2}})^{\frac{N_s}{2}}} \text{exp} \bigg\{ -\frac{1}{2{\sigma_{n_1}^2}} ({\bold{y}}-{\bold{K}}{\bold{x}})^H({\bold{y}}-{\bold{K}}{\bold{x}}) \bigg\}.
\end{equation}
From ($\ref{eq13a}$) and ($\ref{eq23a}$), it is straightforward to see that the ``regularity conditions" are satisfied, and hence for such a linear estimator, we can write the expression for the CRLB as $\cite{Kay}$
\begin{equation}\label{eq24a}
{\bold{I}^{-1}({\bold{\hat{x}}})} = ({\bold{K}^H}{\bold{C}^{-1}}{\bold{K}})^{-1},
\end{equation}
where C is the noise covariance matrix of ${\bold{n_1}}$ as given in ($\ref{eq22a}$).
Substituting $\bold{K}$ and $\bold{C}$ in ($\ref{eq24a}$), we arrive at
\begin{equation}\label{eq26a_appx}
{\bold{I}^{-1}({\bold{\hat{x}}})} = ({\bold{K}^H}{\bold{C}^{-1}}{\bold{K}})^{-1}={\sigma_n^2}{\bold{\Sigma}^{-2}}+{\bold{K}^{-1}}{\bold{W}_D^H}{\bold{D}_q^2}{\bold{W}_D}({\bold{K}^H})^{-1}.
\end{equation}
Now, on setting ${\bold{K}} = {\bold{I}_{N_s}}$, the expression for the CRLB is simplified to 
\begin{equation}\label{eq27a}
{\bold{I}^{-1}({\bold{\hat{x}}})} = {\sigma_n^2}{\bold{\Sigma}^{-2}}+{\bold{W}_D^H}{\bold{D}_q^2}{\bold{W}_D}.\end{equation}
Thus, for a given ${\bold{W}_A^H}$, ${\bold{W}_D^H}$ and $\bold{W}_\alpha$, we see that the expression for CRLB in ($\ref{eq27a}$) is the same as the $\mbox{MSE}(\bold{x})$ in ($\ref{10ac}$). Hence, $\mbox{MSE}(\bold{x})$ in ($\ref{10ac}$) is indeed MMSE.
\end{proof}
\begin{theorem}\label{Thm2}
If $\bold{n}_1 = {\bold{W}_D^H}{\bold{W}_{\alpha}}{\bold{W}_A^H}{\bold{n}} + {\bold{W}_D^H}{\bold{n}_q}$, where $\bold{n}$ is $\bold{n} \sim \mathcal{CN}(\bold{0},{\sigma_n^2\bold{I}_{N_s}})$ and ${\bold{n}_q} \sim \mathcal{N}(\bold{0},{\bold{D}_q^2})$ with ${\bold{D}_q^2} = {\bold{W}_{\alpha}}{\bold{W}_{1-\alpha}}{\text{diag}}[ {\bold{W}_A^H}{\bold{H}}({\bold{W}_A^H}{\bold{H}})^H+{\bold{I}_{N_s}}]$, then it can be shown that $\bold{n}_1$ is circularly symmetric complex Gaussian (CSCG) vector. That is, $\bold{n_1} \sim \mathcal{CN}(\bold{0},\bold{\Phi})$.
\end{theorem}
\begin{proof}
The condition for the random vector $\bold{n}_1$ to be CSCG is \cite{RobertGallager}
\begin{equation}\label{proof_eq1}
E[\bold{n}_1] = E[\bold{n}_1\bold{n}_1^T] = \bold{0}.
\end{equation}
Here, $E[\bold{n}_1\bold{n}_1^T]$ is the pseudo covariance. We first prove that $\bold{n}_q$ is CSCG distributed as $\bold{n}_q \sim \mathcal{N}(\bold{0},{\bold{D}_q^2})$. Given ${\bold{D}_q^2} = E[\bold{n}_q\bold{n}_q^H] = {\bold{W}_{\alpha}}{\bold{W}_{1-\alpha}}{\text{diag}}[ {\bold{W}_A^H}{\bold{H}}({\bold{W}_A^H}{\bold{H}})^H+{\bold{I}_{N_s}}]$; with $\bold{W}_{\alpha}$, $\bold{W}_{1-\alpha}$ and ${\text{diag}}[ {\bold{W}_A^H}{\bold{H}}({\bold{W}_A^H}{\bold{H}})^H+{\bold{I}_{N_s}}]$ being positive real diagonal matrices, effectively results in the covariance matrix ${\bold{D}_q^2}$ being positive real diagonal.\\
A necessary and sufficient condition \cite{Vishwa, RobertGallager} for a random vector $\bold{n}_q$ to be a CSCG random vector is that it has the form $\bold{n}_q = \bold{A}\bold{w}$ where $\bold{w}$ is iid complex Gaussian, that is $\bold{w} \sim \mathcal{CN}(\bold{0},\bold{I_{N_s}})$ and $\bold{A}$ is an arbitrary complex matrix.
Since ${\bold{D}_q^2}$ is positive real diagonal matrix, we can express
\begin{equation}\label{proof_eq2}
\bold{n}_q = \bold{D}_q\bold{w},
\end{equation}
where $\bold{w} \sim \mathcal{CN}(\bold{0},\bold{I_{N_s}})$. This leads to $E[\bold{n}_q] = \bold{D}_qE[\bold{w}] = \bold{0}$ and $E[\bold{n}_q\bold{n}_q^T] = \bold{D}_qE[\bold{w}\bold{w}^T]\bold{D}_q = \bold{0}$. Hence $\bold{n}_q$ is circularly symmetric jointly Gaussian random vector. $\bold{n}_q \sim \mathcal{CN}(\bold{0},{\bold{D}_q^2})$.\\
\noindent
Using $\ref{proof_eq2}$, we can express $\bold{n}_1$ as
\begin{equation}\label{proof_eq4}
\begin{split}
\bold{n}_1 = {\bold{W}_D^H}{\bold{W}_{\alpha}}{\bold{W}_A^H}{\bold{n}} + {\bold{W}_D^H}\bold{D}_q{\bold{w}}
\end{split}
\end{equation}
Since we have $\bold{n}$ and $\bold{w}$ as i.i.d complex Gaussian vectors, we can write
\begin{equation}\label{proof_eq5}
\begin{split}
&E[\bold{n}\bold{n}^T] = E[\bold{w}\bold{n}^T] = E[\bold{n}\bold{w}^H] = E[\bold{w}\bold{n}^H] = \bold{0},\\
&E[\bold{n}\bold{n}^H] = \sigma_n^2\bold{I}_{N_s},\\
&E[\bold{w}\bold{w}^H] = \bold{I}_{N_s}.\\
\end{split}
\end{equation}
Thus, we arrive at
\begin{equation}\label{proof_eq6}
\begin{split}
E[\bold{n}_1] &= {\bold{W}_D^H}{\bold{W}_{\alpha}}{\bold{W}_A^H}E[{\bold{n}}] + {\bold{W}_D^H}\bold{D}_qE[{\bold{w}}] = 0.\\
E[\bold{n}_1\bold{n}_1^T] &= \bold{G}E[\bold{n}\bold{n}^T]\bold{G}^T + \bold{G}E[\bold{n}\bold{w}^T]\bold{D}_q\bold{W}_D + \bold{W}_D^T\bold{D}_qE[\bold{w}\bold{n}^T]\bold{G}^T + \bold{W}_D^T\bold{D}_qE[\bold{w}\bold{w}^T]\bold{D}_q\bold{W}_D=\bold{0}.
\end{split}
\end{equation}
Also, 
\begin{equation}\label{proof_eq7}
\begin{split}
&E[\bold{n}_1\bold{n}_1^H] = \bold{G}E[\bold{n}\bold{n}^H]\bold{G}^H+\bold{G}E[\bold{n}\bold{w}^H]\bold{D}_q\bold{W}_D+ \bold{W}_D^H\bold{D}_qE[\bold{w}\bold{n}^H]\bold{G}^H+\bold{W}_D^H\bold{D}_qE[\bold{w}\bold{w}^H]\bold{D}_q\bold{W}_D,\\
&E[\bold{n}_1\bold{n}_1^H] = \bold{\Phi} = \sigma_n^2\bold{G}\bold{G}^H + \bold{W}_D^H\bold{D}_q^2\bold{W}_D.
\end{split}
\end{equation}
Thus, ${\bold{n}_1} \sim {\mathcal{CN}}({\bold{0}},{\bold{\Phi}})$ is a CSCG vector.
\end{proof}
\newtheorem{lem}{Lemma}
\begin{lem}\label{lemm1}
The term  $\log_2  \Big( q(b_i) + 1 \Big)$ for $0 \leq q(b_i) < 1$, can be approximated as $\log_2  \Big( q(b_i) + 1 \Big) \simeq \frac{q(b_i)}{\ln2}$. 
\end{lem}
\begin{proof}
We can write:\\
$\log_2  \Big( \frac{p\sigma_i^2}{\sigma_n^2 + g(b_i)l_i} + 1 \Big) = \frac{\ln \Big( \frac{p\sigma_i^2}{\sigma_n^2 + g(b_i)l_i} + 1 \Big)}{\ln 2}$.\\
\indent
We can approximate $g(b_i)$ as $c2^{-db_i}$, where $d=2.0765, c=2.40667$. For the sake of simplicity, we will replace the variable $\bold{b} \in \mathbb{I}^{N_s \times 1}$ with $\bold{x} \in \mathbb{R}^{N_s \times 1}$.\\
\indent
We will now define $f\big(p(x_i)\big) = \ln \Big( \frac{p\sigma_i^2}{\sigma_n^2 + c2^{dx_i}l_i} + 1 \Big)$, where $p(x_i) = \frac{p\sigma_i^2}{\sigma_n^2 + c2^{dx_i}l_i}$. For a geometric series below, which has a common ratio of $-p(x_i)$, where $0 \leq p(x_i) < 1$, we can write
\begin{equation}\label{apx_c1}
1-p(x_i)+p(x_i)^2-p(x_i)^3+.. = \frac{1}{1+p(x_i)}.
\end{equation}
\begin{equation}\label{apx_c2}
\ln(1 + p(x_i)) = \int \frac{1}{1 + p(x_i)} d(p(x_i)),
\end{equation}
substituting for $\frac{1}{1+p(x_i)}$ into the integral in $\ref{apx_c2}$ from $\ref{apx_c1}$, we have
\begin{equation}\label{apx_c3}
\ln(1 + p(x_i)) = p(x_i)-\frac{p(x_i)^2}{2}+\frac{p(x_i)^3}{3}-\frac{p(x_i)^4}{4}+...
\end{equation}
Since we know that $0 \leq p(x_i) < 1$, the higher powers of $p(x_i)$ are negligible and thus the above series can be approximated  as
\begin{equation}\label{apx_c4}
f\big(p(x_i)\big)  \simeq p(x_i).
\end{equation}
By re-substituting variable $\bold{x} \in \mathbb{R}^{N_s \times 1}$ with $\bold{b} \in \mathbb{I}^{N_s \times 1}$, we can effectively write
\begin{equation}\label{apx_c5}
\log_2  \Big( \frac{p\sigma_i^2}{\sigma_n^2 + g(b_i)l_i} + 1 \Big) \simeq \frac{1}{\ln2}\Big(\frac{p\sigma_i^2}{\sigma_n^2 + g(b_i)l_i}\Big). 
\end{equation}
\end{proof}
\vspace{-13mm}
\begin{lem}\label{lemm2}
It can be shown that $\log_2  \Big( q(b_i) + 1 \Big) = \Big(1-\frac{1}{q(b_i)}\Big)P + L(p,\sigma_i^2, \sigma_n^2)$ for $\infty > q(b_i) \geq 1$, where the terms $P$ and $L(p,\sigma_i^2, \sigma_n^2)$ are not functions of $b_i$.
\end{lem}
\begin{proof}
Consider the expansion for $f\big(p(x_i)\big)$ for $\infty > p(x_i) \geq 1$. We can approximate $f\big(p(x_i)\big)$ as
\begin{equation}\label{apx0}
f(\big(p(x_i)\big) = \ln \Big( p(x_i) + 1 \Big) \simeq \ln \Big(p(x_i) \Big).
\end{equation}
Rewriting $f(\big(p(x_i)\big)$ as:
\begin{equation}\label{apx1}
\begin{split}
&f(\big(p(x_i)\big) = -\ln \bigg( \frac{1}{p(x_i)} \bigg) \text{ for } 0 < \frac{1}{p(x_i)} \leq 2;\\
&f(\big(p(x_i)\big) = -\ln \Big(g(x_i)\Big) \text{ where } g(x_i) = \frac{1}{p(x_i)};\\
\text{ or }& f(\big(p(x_i)\big) = -h\big(g(x_i)\big) \text{ where } h\big(g(x_i)\big) = \ln\big(g(x_i)\big);
\end{split}
\end{equation}
Evaluating the Taylor series at $g(x_i = x_0) = 1 = \frac{1}{p(x_i = x_0)}$ with the region of convergence $R: \infty > p(x_i) \geq \frac{1}{2}$, we have
\begin{equation}\label{apx2}
h\big(g(x_i)\big) = h\big(g(x_0)\big) + h^\prime\big(g(x_0)\big)(g(x_i)-1) + \frac{1}{2}h^{\prime\prime}\big(g(x_0)\big)(g(x_i)-1)^2 + \frac{1}{6}h^{\prime\prime\prime}\big(g(x_0)\big)(g(x_i)-1)^3 + ..
\end{equation}
Also:
\begin{equation}\label{apx3}
\begin{split}
h\big(g(x_0)\big) &= \ln(1) = 0; \text{  }h^\prime\big(g(x_i)\big) = \frac{1}{g(x_i)} \implies h^\prime\big(g(x_0)\big) = 1;\\
h^{\prime\prime}\big(g(x_i)\big) &= -\frac{1}{[g(x_i)]^2},  h^{\prime\prime}\big(g(x_0)\big) = -1; \text{  }h^{\prime\prime\prime}\big(g(x_i)\big) = \frac{2}{[g(x_i)]^3}, h^{\prime\prime\prime}\big(g(x_0)\big) = 2;\\
\vdots
\end{split}
\end{equation}
substituting $\ref{apx3}$ in $\ref{apx2}$, we have
\begin{equation}\label{apx4}
\begin{split}
h\big(g(x_i)\big) = &\bigg(\frac{1}{p(x_i)}-1\bigg) - \frac{1}{2}\bigg(\frac{1}{p(x_i)}-1\bigg)^2+\frac{1}{3}\bigg(\frac{1}{p(x_i)}-1\bigg)^3 - ..\\
f\big(p(x_i)\big)= & \bigg(1-\frac{1}{p(x_i)}\bigg) - \sum_{n=2}^{\infty} \frac{(-1)^{(n-1)}}{n}\bigg(\frac{1}{p(x_i)}-1\bigg)^n
\end{split}
\end{equation}
Using binomial expansion for $\Big(\frac{1}{p(x_i)}-1\Big)^n$, we can write
\begin{equation}\label{apx6a}
\bigg(\frac{1}{p(x_i)}-1\bigg)^n = \sum_{k=0}^n {{n}\choose{k}}\frac{-1^{(n-k)}}{(p(x_i))^k} = K_n(p,\sigma_i^2, \sigma_n^2).
\end{equation}
It is to be noted that for $n \ge 2$ and larger values of $k$, the function $K_n(p,\sigma_i^2, \sigma_n^2)$ becomes independent of $x_i$ and is convergent for $p(x_i) \ge 1$. So, we can write $\ref{apx6a}$ safely as
\begin{equation}\label{apx5}
\begin{split}
f\big(p(x_i)\big)= & \bigg(1-\frac{1}{p(x_i)}\bigg) - \sum_{n=2}^{\infty} \frac{(-1)^{(n-1)}K_n(p,\sigma_i^2, \sigma_n^2)}{n}\\
=&\bigg(1-\frac{1}{p(x_i)}\bigg) + G(p,\sigma_i^2, \sigma_n^2),
\end{split}
\end{equation}
Where $G(p,\sigma_i^2, \sigma_n^2) = - \sum_{n=2}^{\infty} \frac{(-1)^{(n-1)}K_n(p,\sigma_i^2, \sigma_n^2)}{n}$ and is a converging series.\\
By re-substituting variable $\bold{x} \in \mathbb{R}^{N_s \times 1}$ with $\bold{b} \in \mathbb{I}^{N_s \times 1}$, we can effectively write
\begin{equation}\label{apx6}
\log_2  \bigg( \frac{p\sigma_i^2}{\sigma_n^2 + g(b_i)l_i} + 1 \bigg) = P\bigg(1-\frac{1}{\frac{p\sigma_i^2}{\sigma_n^2 + g(b_i)l_i}}\bigg) + L(p,\sigma_i^2, \sigma_n^2).
\end{equation}
where $P = \frac{1}{\ln2}$ and $L(p,\sigma_i^2, \sigma_n^2) = \frac{G(p,\sigma_i^2, \sigma_n^2)}{\ln2}$.
\end{proof}

\section*{Acknowledgment}
The authors would like to thank National Instruments for the financial support extended for this work.  



%
\bibliographystyle{IEEEtran}
\bibliography{Opt_BA_MaMIMO_arxiv}
\end{document}